\RequirePackage{amsthm}
\documentclass[sn-mathphys-num,pdflatex]{sn-jnl}
\usepackage{multirow}
\usepackage{graphicx}%
\usepackage{multirow}%
\usepackage{amsmath,amssymb,amsfonts}%
\usepackage{amsthm}%
\usepackage{mathrsfs}%
\usepackage[title]{appendix}%
\usepackage{xcolor}%
\usepackage{textcomp}%
\usepackage{manyfoot}%
\usepackage{booktabs}%
\usepackage{algorithm}%
\usepackage{algorithmicx}%
\usepackage{algpseudocode}%
\usepackage{listings}%

\usepackage{overpic}
\usepackage[h]{esvect}
\usepackage{diagbox}

\theoremstyle{thmstyleone}%
\newtheorem{theorem}{Theorem}
\newtheorem{corollary}[theorem]{Corollary}%
\newtheorem{lemma}[theorem]{Lemma}%
\newtheorem{construction}{Construction}

\theoremstyle{thmstyletwo}%
\newtheorem{example}{Example}%

\theoremstyle{thmstylethree}%
\newtheorem{definition}{Definition}%

\usepackage{mathtools}
\DeclarePairedDelimiter\abs{\lvert}{\rvert}

\DeclarePairedDelimiter\ceil{\lceil}{\rceil}
\DeclarePairedDelimiter\floor{\lfloor}{\rfloor}
\DeclarePairedDelimiter\parenv{\lparen}{\rparen}
\DeclarePairedDelimiter\sparenv{\lbrack}{\rbrack}

\DeclarePairedDelimiter\angenv{\langle}{\rangle}
\DeclarePairedDelimiter\set{\{}{\}}
\DeclarePairedDelimiter\mset{\{\!\!\{}{\}\!\!\}}

\renewcommand{\leq}{\leqslant}

\renewcommand{\geq}{\geqslant}

\newcommand{\cA}{\mathcal{A}}
\newcommand{\cB}{\mathcal{B}}

\newcommand{\cH}{\mathcal{H}}

\newcommand{\N}{\mathbb{N}}

\newcommand{\eqdef}{\triangleq}

\DeclareMathOperator{\cost}{\mathsf{cost}}

\newcommand{\ba}{\boldsymbol{a}}

\newcommand{\bi}{\boldsymbol{i}}

\newcommand{\bx}{\boldsymbol{x}}

\newcommand{\umod}{\,\overline{\bmod}\,}

\begin{document}

\title[]{On Zero Skip-Cost Generalized Fractional-Repetition Codes from Covering Designs}


\author*[1]{\fnm{Wenjun} \sur{Yu}}\email{wenjun@post.bgu.ac.il}

\author[2]{\fnm{Bo-Jun} \sur{Yuan}}\email{ybjmath@163.com}

\author[1,3]{\fnm{Moshe} \sur{Schwartz}}\email{schwartz.moshe@mcmaster.ca}

\affil[1]{\orgdiv{School of Electrical and Computer Engineering}, \orgname{Ben-Gurion University of the Negev}, \orgaddress{\city{Beer Sheva}, \postcode{8410501}, \country{Israel}}}

\affil[2]{\orgdiv{School of Science}, \orgname{Zhejiang University of Science and Technology}, \orgaddress{\city{Hangzhou}, \postcode{310023}, \country{China}}}

\affil[3]{\orgdiv{Department of Electrical and Computer Engineering}, \orgname{McMaster University}, \orgaddress{\city{Hamilton}, \postcode{L8S 4K1}, \state{ON}, \country{Canada}}}


\abstract{
We study generalized fractional repetition codes that have zero skip cost, and which are based on covering designs. We show that a zero skip cost is always attainable, perhaps at a price of an expansion factor compared with the optimal size of fractional repetition codes based on Steiner systems. We provide three constructions, as well as show non-constructively, that no expansion is needed for all codes based on sufficiently large covering systems.
}




\maketitle

\section{Introduction}

DRESS (Distributed Replication based Exact Simple Storage) codes were introduced by El Rouayheb and Ramchandran~\cite{el2010fractional}, generalizing a construction of Rashmi et al.~\cite{rashmi2009explicit}. These form a class of exact minimum-bandwidth regenerating (MBR) codes for distributed storage systems that offer low-complexity uncoded repair of failed nodes. In this scheme, user information is first encoded using an outer MDS code. The coded information is then further coded using an inner fractional-repetition (FR) code, which distributes copies of the coded information across several nodes.

When a node fails (called, node erasure), a small number of helper nodes are contacted. These nodes transfer copies of parts of the erased data, thereby recovering all of the erased node. The number of helper nodes (called, code locality) and their identity are determined by the FR code.

Since constructions for MDS codes with flexible parameters are known, considerable effort has been devoted to the construction of FR codes. A common strategy for the latter is to base the construction on combinatorial designs (e.g.,~\cite{olmez2016fractional,el2010fractional,zhu2018study,zhu2019duality,olmez2012repairable,zhu2019fractional}). Among these, using Steiner systems is perhaps the most popular. However, since only few constructions of Steiner systems are known, and their parameters are quite restrictive, it was recently suggested to consider other designs~\cite{zhu2014general,zhu2017general} resulting in generalized fractional-repetition (GFR) codes.

When considering the efficiency of recovering from a node erasure, it is not only the number of helper nodes involved (locality) and the amount of data read in each helper node (I/O cost), but also the way it is arranged within the helper node~\cite{wu2021achievable}. In particular, when a helper node reads the required information for an erasure recovery process, the number of elements it needs to skip adversely affects the recovery speed. This motivated Chee et al.~\cite{chee2024repairing} to define the skip cost of the FR code, with zero skip cost being the most desirable. They went on to construct zero skip-cost FR codes based on Steiner quadruple systems. Apart from this, no other non-trivial FR codes with zero skip cost are known.

The main contributions of this paper are as follows: We construct generalized fractional repetition (GFR) codes with zero skip cost based on covering designs, which we call \emph{covering fractional repetition (CFR)} codes. Covering designs are known to exist for all parameters while asymptotically approaching those of Steiner systems, as shown by R\"odl~\cite{rodl1985packing}. The zero skip cost, however, comes at a price of potentially having more nodes than an FR code based on a Steiner system of the same parameters. We call the ratio of obtained nodes to the optimum number of nodes, the expansion factor. We show a trivial construction with asymptotic expansion factor $2$, and then a second construction with asymptotic expansion factor $1$. However, in finite cases, the expansion factor of the second construction might be higher than that of the first. We then provide a third more efficient construction for finite cases, improving upon the first two in many cases. Finally, we show using a probabilistic non-constructive argument, that for all sufficiently-large covering designs, no expansion is required whatsoever in order to attain zero skip cost.

The paper is organized as follows. In Section~\ref{sec:prelim} we give the notation and rigorous definitions used throughout the paper. Our constructions of CFR codes are described in Section~\ref{sec:cfr}. We conclude in Section~\ref{sec:conc} with a short summary and open questions.

\section{Preliminaries}
\label{sec:prelim}

For integers $i<j$, we define $[i,j]\eqdef\set{i,i+1,\dots,j}$, and for $v\in \N$, we define $[v] \eqdef [1,v]$. Given a set $S$, we denote by $\binom{S}{k}$ the set of all $k$-subsets of $S$, namely,
\[ \binom{S}{k} \eqdef \set*{S'\subseteq S ~:~ \abs*{S'}=k}.\]
For any two positive integers, $a$ and $b$, we let $a\umod b$ denote the unique integer $r\in [b]$ such that $r\equiv a \pmod{b}$. Thus,
\[
a = \parenv*{\ceil*{\frac{a}{b}}-1} b + (a \umod b).
\]
We shall also use the notation $\mset{}$ to denote multisets, i.e., sets with possibly repeated elements.


Throughout this paper we work over the alphabet $[v]$, and we consider the distributed storage system based on a $k\times N$ array $\cA\in [v]^{k\times N}$ with form $\cA = (\ba^{(1)},\dots,\ba^{(N)})$, where $\ba^{(i)}\in [v]^k$ is column vector of length $k$. In this framework, each column of $\cA$ represents a node. There are $v$ symbols spread among the $N$ nodes such that the $i$-th node stores $k$ symbols in an ordered sequence $\ba^{(i)}$. 

When some node $j$ fails, the system repairs it by transferring symbols from other nodes. We consider \emph{repair by transfer}, also called \emph{uncoded repair}~\cite{rashmi2009explicit}. The repair process calls upon a set of nodes, called \emph{helper nodes}, $\cH^{(j)}\subseteq [N]\setminus \set{j}$, which depends on $j$. Each node $h\in\cH^{(j)}$ transmits a subset of its entries, $R^{(h)}$, back to node $j$ such that the original entries of $j$ are all received, namely, $\bigcup_{h\in\cH^{(j)}} R^{(j)}$ is exactly the set of entries in $\ba^{(j)}$. Since we would like to minimize the bandwidth required for repairing, we further assume no element is transmitted twice, which means that $R^{(h)}\cap R^{(h')}=\emptyset$, for all $h,h'\in\cH^{(j)}$, $h\neq h'$. Another standard assumption we make is that the repair process has \emph{locality} $\ell$, which means no more than $\ell$ helper nodes participate in the repair, i.e., $\abs{\cH^{(j)}}\leq \ell$.

At this point we would like to consider the \emph{skip cost} of the repair process. The concept was first introduced by~\cite{chee2024repairing} as a measure of the repair-process performance. It is motivated by storage devices that are more efficient reading sequentially (e.g., HDDs and magnetic tapes) compared with random access that requires skipping.

\begin{definition}
Assume the distributed-storage code $\cA$, as defined above, with locality $\ell$. Further assume that during a repair process, node $h\in[N]$ transmits a subset of its contents, $R^{(h)}$, where
\begin{align*}
    \ba^{(h)}&=(a^{(h)}_1,a^{(h)}_2,\dots,a^{(h)}_k), \\
    R^{(h)}&=\set*{a^{(h)}_{i_1},a^{(h)}_{i_2},\dots,a^{(h)}_{i_t}},
\end{align*}
for some $1 \leq i_1 < \ldots < i_t \leq k$. the \emph{skip cost} for node $h$ is defined in this case as
\[
\cost(R^{(h)} \mid \ba^{(h)}) \eqdef \sum_{j=1}^{t-1} (i_{j+1}-i_j-1) = i_t - i_1 - (t-1).
\]
We naturally extend this to the skip cost of repairing node $j$,
\[\cost(\ba^{(j)}) \eqdef \min_{\cH^{(j)}}\sum_{h\in \cH^{(j)}} \cost(R^{(h)} \mid \ba^{(h)}),\]
where the minimization is over all ways of finding a set of helper nodes $\cH^{(j)}$, $\abs{\cH^{(j)}}\leq \ell$, and assigning each $h\in\cH^{(j)}$ a transmission $R^{(h)}$. Finally, the skip cost of the entire array is defined by
\[\cost(\cA) \eqdef \max_{j\in [N]} \cost(\ba^{(j)}).\] 
\end{definition}

\begin{example}
Consider a distributed-storage system based on the following array,
\[
\cA = 
\begin{pmatrix}
       1 & 2 & 1 & 2 & 1 & 1 \\
       2 & 3 & 3 & 4 & 3 & 2 \\
       3 & 4 & 4 & 5 & 5 & 4 \\
       5 & 6 & 5 & 6 & 6 & 6 \\
    \end{pmatrix},
\]
and assume locality $\ell=2$.

If node $1$ fails, there are some repair options. One option is to contact nodes $\cH^{(1)}_1=\set{3,6}$, and have node $3$ transmit $\set{3,5}$ and node $6$ transmit $\set{1,2}$. The skip cost of this option is
\[
\cost(\set{3,5}\mid \ba^{(3)})+ \cost(\set{1,2} \mid \ba^{(6)}) = 1 + 0 = 1.
\]
We can alternatively contact nodes $\cH^{(1)}_2=\set{5,6}$, and have node $5$ transmit $\set{3,5}$ and node $6$ transmit $\set{1,2}$. The skip cost of this option is
\[
\cost(\set{3,5}\mid \ba^{(5)})+ \cost(\set{1,2} \mid \ba^{(6)}) = 0 + 0 = 0.
\]
We naturally have
\[ \cost( \ba^{(1)}) = 0,\]
since the skip cost is always non-negative and we already have evidence of a repair process for node $1$ with skip cost $0$.

We can summarize the lowest skip cost repair process for each of the columns of $\cA$ in the following table:
\[
\begin{array}{ccc}
\toprule
\text{Column} & \text{Repair Process} & \text{Skip Cost} \\
\midrule
1 & \set{1,2}\text{ from } \ba^{(6)}, \set{3,5}\text{ from } \ba^{(5)} & 0 \\
2 & \set{2,3}\text{ from } \ba^{(1)}, \set{4,6}\text{ from } \ba^{(6)} & 0 \\
3 & \set{1,3}\text{ from } \ba^{(5)}, \set{4,5}\text{ from } \ba^{(4)} & 0 \\
4 & \set{2,4}\text{ from } \ba^{(6)}, \set{5,6}\text{ from } \ba^{(5)} & 0 \\
5 & \set{1,3}\text{ from } \ba^{(3)}, \set{5,6}\text{ from } \ba^{(4)} & 0 \\
6 & \set{1,2}\text{ from } \ba^{(1)}, \set{4,6}\text{ from } \ba^{(2)} & 0 \\
\bottomrule
\end{array}
\]
Since all columns may be repaired with a skip cost of $0$, we can now say
\[ \cost(\cA) = 0.\]
\end{example}

To better understand the considerations in constructing $\cA$, we recall the framework of fractional repetition codes. El Rouayheb and Ramchandran~\cite{el2010fractional} introduced the concept of a DRESS (Distributed Replication-based Exact Simple Storage) code, which consists of the concatenation of an outer MDS code and an inner fractional repetition (FR) code. In their setting\footnote{We use different letters than~\cite{el2010fractional} for the code parameters, to better align with the common notation in design theory.}, assume that $k,N,v,\rho$ are positive integers satisfying
\begin{equation}
\label{eq:knvrho}
kN=v\rho.
\end{equation}
An $(N,k,\rho)$-FR code is a collection of $N$ $k$-subsets of $[v]$, $B_1,B_2,\dots,B_N\in\binom{[v]}{k}$, such that each symbol of $[v]$ appears in \emph{exactly} $\rho$ subsets. For convenience, a $k\times N$ matrix $\cA$ is defined, whose $i$-th column, $\ba^{(i)}$, contains the elements of $B_i$ in some arbitrary order.


Several studies have suggested ways of constructing the blocks $B_1,\dots,B_n$ which define the FR code. Most of these revolve around the use block designs, as these naturally contain points and blocks as part of their definition. In some cases the blocks of the system are used as is, while in some the transpose design is used, where the role of points and blocks is reversed. As examples we mention the use of Steiner systems in~\cite{olmez2016fractional,el2010fractional,zhu2018study,chee2024repairing}, $t$-$(v,k,\lambda)$ designs in~\cite{zhu2019duality}, and affine resolvable designs~\cite{olmez2012repairable,olmez2016fractional,zhu2019fractional}.

The popularity of Steiner systems is shadowed by their relative rarity. The strict divisibility conditions required prohibit many parameters, and the few constructions known provide only Steiner systems with small parameters. However, in~\cite{zhu2014general,zhu2017general}, more flexible codes were suggested, called \emph{generalized fractional repetition (GFR) codes}, paving the way to using other designs, such as covering designs, which are the focus of this paper. To describe these codes, we first recall the definition of covering designs.

\begin{definition}
Let $t \leq k \leq v$ be positive integers. A \emph{$(t,k,v)$ covering design} is the pair $(X,\cB)$, where $X=[v]$ (the \emph{points}) and a multiset $\cB=\mset{B_1,\ldots,B_N}\subseteq \binom{X}{k}$ (the \emph{blocks}), such that any $t$-subset of points $T\in \binom{X}{t}$ is contained in at least one block. 
\end{definition}

This definition generalizes the notion of $(t,k,v)$ Steiner systems, which require that any $t$-subset is contained in \emph{exactly one} block.

Assume $(X,\cB)$ is a $(t,k,v)$ covering design. Let us then define
\begin{equation}
\label{eq:rs}
r_s\eqdef \frac{\binom{v-s}{t-s}}{\binom{k-s}{t-s}},
\end{equation}
for all $s\in[t]$. If we are given an $s$-subset, $S\in\binom{X}{s}$, we use $r_S$ to denote the number of blocks containing $S$. By standard counting techniques (e.g., see~\cite{stinson2008combinatorial}), we have
\begin{equation}
\label{eq:rsineq}
r_S \geq \frac{\binom{v-s}{t-s}}{\binom{k-s}{t-s}}=r_s.
\end{equation}
For $(t,k,v)$ Steiner systems, \eqref{eq:rsineq} holds with equality. Since $k\leq v$,
\begin{equation}
\label{eq:rshierarchy}
r_1\geq r_2 \geq \cdots \geq r_{t-1} \geq r_t.
\end{equation}
Additionally, the number of blocks satisfies,
\begin{equation}
\label{eq:minB}
\abs{\cB} \geq \frac{\binom{v}{t}}{\binom{k}{t}},
\end{equation}
and once again, this is attained with equality for Steiner systems.

Let us now define GFR codes based on covering designs.

\begin{definition}
\label{def:cfr}
Let $(X,\cB)$ be a $(t,k,v)$ covering design with blocks $\cB=\mset{B_1,\dots,B_n}$. A \emph{covering fractional repetition (CFR) code} with parameters $(N=\abs{\cB},k,\rho)_v$ is a $k\times N$ array $\cA\in[v]^{k\times N}$, whose columns are the blocks of $\cB$ in some order. Each symbol of $[v]$ appears in at least $\rho$ columns of $\cA$.
\end{definition}

Thus, a CFR code is just the blocks of a covering design written as columns of an array. The order of the columns does not matter. Within each column we may choose to order the block elements in any way we want, however, this crucially determines the skip cost of the code.

The main relaxation of CFR codes, compared with FR codes, is that not every symbol of $[v]$ appears the same number of times in the array $\cA$. We can always safely assume $\rho$ from Definition~\ref{def:cfr} is at least $r_1$ from~\eqref{eq:rs}. If we denote by $\rho_a$ the number of time $a\in[v]$ appears in $\cA$, then we have
\[
kN = \sum_{a\in[v]}\rho_a,
\]
which reduces to~\eqref{eq:knvrho} when all the $\rho_a$ are equal. Thus, each node contains exactly $k$ symbols, and each symbol $a\in[v]$ appears in $\rho_a$ nodes. We comment that~\cite{zhu2014general,zhu2017general} study the transpose design, resulting in each symbol in $[v]$ appearing the same number of times, but nodes may contain a different number of symbols.

Finally, we recall R\"odl's proof of the existence of good covering designs~\cite{rodl1985packing}:

\begin{lemma}[\cite{rodl1985packing}]
\label{lem: Rodl}
For all positive integers $t \leq k \leq v$, there exist $(t,k,v)$ covering designs $(X,\cB)$, with size 
\[\abs*{\cB}=\frac{\binom{v}{t}}{\binom{k}{t}}(1+o(1)),\]
when $t$ and $k$ are constants, and $v\to\infty$. 
\end{lemma}

Lemma~\ref{lem: Rodl} shows that it is possible to build covering designs whose number of blocks approaches that of Steiner systems, even in parameters where no Steiner systems exist.

\section{Constructions of Zero Skip-Cost CFR Codes}
\label{sec:cfr}

The goal of this section is to construct CFR codes with zero skip cost. As we show shortly, this is always possible in a trivial way -- simply repeat columns in the CFR code. However, we would also like to minimize the number of columns in the CFR code, a notion we shall also formalize as the \emph{expansion factor}. Another simple construction we present gives an asymptotic optimal expansion factor, but poor expansion factor for many practical parameters. We then present a more elaborate construction which significantly reduces the expansion factor. We conclude by showing, using the probabilistic method, that in fact almost all orderings of covering designs have zero skip cost, and no expansion is needed. This, unfortunately, is not a constructive proof.

We first ask ourselves what is the locality restriction. Particularly, since lower locality is better, we shall determine the minimum locality we can guarantee given just the parameters of the covering design used to construct the CFR code. To that end we first define a set of desired design parameters.

\begin{definition}
Let $(X,\cB)$ be a $(t,k,v)$ covering design. We say it is \emph{properly local} if
\begin{equation}
\label{eq:proploc}
r_{t-1}=\ceil*{\frac{v-t+1}{k-t+1}}\geq \ceil*{\frac{k}{t-1}}+1.
\end{equation}
\end{definition}

\begin{lemma}
\label{lem:locality}
Let $\cA\in[v]^{k\times N}$ be a CFR code based on a properly local $(t,k,v)$ covering design $(X,\cB)$. Then the minimum locality of $\cA$ satisfies
\[
\ell_{\min} \leq \ceil*{\frac{k}{t-1}}.
\]
\end{lemma}
\begin{proof}
Let us define
\begin{align*}
q&\eqdef \ceil*{\frac{k}{t-1}}, & r &\eqdef k \umod (t-1).
\end{align*}
Assume a column of $\cA$ is erased containing the elements of the block $\set{x_1,\dots,x_k}$. We partition these elements into $q$ disjoint sets \[S_1\eqdef\set{x_1,\dots,x_{t-1}},\quad S_2\eqdef\set{x_t,\dots,x_{2(t-1)}},\quad \dots\quad S_q\eqdef\set{x_{(q-1)(t-1)+1},\dots,x_k},\]
all of size $t-1$, except perhaps the last set of size $r$. We choose the helper nodes iteratively. In the first round, by~\eqref{eq:rs} and~\eqref{eq:rsineq}, $S_1$ is contained in  
\[r_{S_1}-1\geq r_{t-1} -1 = \ceil*{\frac{v-t+1}{k-t+1}} -1
\]
distinct non-erased columns of $\cA$. We choose one arbitrarily as a helper node. In the second round, once again, we have $r_{S_2}-1\geq r_{t-1}-1$ non-erased columns containing $S_2$. If a previously chosen helper node is in this set, we remove it, and choose a helper node for recovering $S_2$ from the remaining ones. We keep this process until reaching $S_q$, which by~\eqref{eq:rshierarchy} also satisfies $r_{S_q}-1\geq r_{t-1}-1$. Since we have a total of $q$ iterations, and since the covering design is properly local (hence $r_{t-1}-1\geq q$ by~\eqref{eq:proploc}), there is at least one column we can choose as a helper node in each iteration. This results in locality $q$, proving the claim on the minimum locality.
\end{proof}

Following Lemma~\ref{lem:locality}, we shall consider only properly local covering designs, and analyze the skip cost of CFR codes based on them using the restriction on the locality to be
\[
\ell = \ceil*{\frac{k}{t-1}}.
\]
We comment that restricting ourselves to properly local covering designs is a very mild restriction, since all $(v,t,k)$ covering designs with $v\geq\frac{k^2}{t-1}$ are properly local, and increasing $v$ for any constant $t$ and $k$ is possible by Lemma~\ref{lem: Rodl}.

The first construction we describe is completely trivial, serving only the purpose of providing a baseline, and motivating the definition of the expansion factor. To describe it, we define the following notation: for any multiset, $S$, and an integer $n\in\N$, we denote by $n\cdot S$ the multiset that contains $n$ copies of each of the elements of $S$.

\begin{construction}
\label{con:multi}
Let $(X,\cB)$ be a properly local $(t,k,v)$ covering design. We construct the CFR code, $\cA$, whose columns are the blocks of $(X,2\cdot\cB)$, in some arbitrary order.
\end{construction}

\begin{theorem}
Let $\cA$ be the CFR code from Construction~\ref{con:multi}. Then it is a $(2\abs{\cB},k,r_1)_v$-CFR code with locality $1$ and $\cost(\cA)=0$, where $r_1$ is given in~\eqref{eq:rs}.
\end{theorem}

\begin{proof}
We first note that if $(X,\cB)$ is a $(t,k,v)$ covering design, then $(X,2\cdot\cB)$ is also a $(t,k,v)$ covering design. Thus, the parameters of $\cA$ follow by definition. If a node is erased, we can recover all of its contents by reading its copy. This immediately gives us locality $1$ and skip cost $0$.
\end{proof}

Construction~\ref{con:multi} is indeed trivial. It also seems quite wasteful, since we simply have repeated nodes. One might therefore ask, given $k$, $\rho$, and $v$ (or alternatively, given $t$, $k$, and $v$), what is the smallest $N$ for which an $(N,k,\rho)_v$-CFR code exists. We observe that~\eqref{eq:minB} gives a lower bound, which inspires the following definition:

\begin{definition}
Let $\cA$ be an $(N,k,\rho)_v$-CFR code based on a $(t,k,v)$ covering design. The \emph{expansion factor} of $\cA$ is defined as
\[
\xi(\cA) \eqdef \frac{N}{\binom{v}{t}/\binom{k}{t}}.
\]
\end{definition}

We have $\xi(\cA)\geq 1$ for all CFR codes, and we would obviously like to have $\xi(\cA)$ as low as possible. 

\begin{corollary}
\label{cor:xi2}
For any fixed $1\leq t\leq k$, there exists an infinite sequence of integers, $k\leq v_1<v_2<\cdots$, such that there are CFR codes, $\cA_i$, with parameters $(N_i,k,\binom{v_i-1}{t-1}/\binom{k-1}{t-1})_{v_i}$, locality $\ell=1$, with $\cost(\cA_i)=0$ and
\[
\lim_{i\to\infty}\xi(\cA_i)=2.
\]
\end{corollary}
\begin{proof}
By Lemma~\ref{lem: Rodl}, there exists a sequence of $(t,k,v_i)$ covering design, $(X_i,\cB_i)$, for which
\[
\lim_{i\to\infty} \frac{\abs*{\cB_i}}{\binom{v_i}{t}/\binom{k}{t}}=1.
\]
The claimed result then follows by using Construction~\ref{con:multi} on this sequence.
\end{proof}

Having constructed zero skip-cost CFR codes with asymptotic expansion factor of $2$, we turn to show a construction with an asymptotic expansion factor of $1$. Unlike the previous construction, this time we shall have to be more specific regarding the ordering of elements in columns.

\begin{construction}\label{con:combination}
Let $(X,\cB)$ be a properly local $(t,k,v)$ covering design, $2\leq t\leq k$, and let $(X,\cB')$ be a $(t-1,k,v)$ covering design. Define $r\eqdef k \umod (t-1)$. 

If $r=t-1$ or $r=1$, we construct the CFR code, $\cA$, whose columns are the blocks of $(X,\cB\cup \binom{k}{t-1}\cdot \cB')$ arranged in the following order:
\begin{itemize}
\item
Blocks from $\cB$ appear in some arbitrary order.
\item
Each block of $\cB'$ is written in $\binom{k}{t-1}$ columns, where each possible $(t-1)$-subset of the block appears contiguously as the first $t-1$ elements in one of the $\binom{k}{t-1}$ copies.
\end{itemize}

If $2\leq r\leq t-2$, we construct the CFR code, $\cA$, whose columns are the blocks of $(X,\cB\cup \binom{k}{t-1}\cdot \cB' \cup \binom{k}{r}\cdot \cB')$ arranged in the following order:
\begin{itemize}
\item
Blocks from $\cB$ appear in some arbitrary order.
\item
Each block of $\cB'$ is written in $\binom{k}{t-1}$ columns, where each possible $(t-1)$-subset of the block appears contiguously as the first $t-1$ elements in one of the $\binom{k}{t-1}$ copies.
\item
Each block of $\cB'$ is written in $\binom{k}{r}$ columns, where each possible $r$-subset of the block appears contiguously as the first $r$ elements in one of the $\binom{k}{r}$ copies.
\end{itemize}
\end{construction}

\begin{theorem}
\label{th:combination}
Let $\cA$ be the CFR code from Construction~\ref{con:combination}. Define
\[
N \eqdef \begin{cases}
\abs{\cB}+\binom{k}{t-1}\abs{\cB'} & r=1,t-1, \\
\abs{\cB}+\parenv*{\binom{k}{t-1}+\binom{k}{r}}\abs{\cB'} & 2\leq r\leq t-2. \\
\end{cases}
\]
Then $\cA$ is a $(N,k,r_1)_v$-CFR code with locality $\ell=\ceil{k/(t-1)}$ and $\cost(\cA)=0$, where $r_1$ is given in~\eqref{eq:rs}.
\end{theorem}

\begin{proof}
The number of columns and rows in $\cA$ is immediate by construction. Since $(X,\cB)$ is a $(t,k,v)$ covering design, adding the copies of blocks from $\cB'$ does not change that fact. Thus, $\rho=r_1$ is guaranteed by~\eqref{eq:rs}. The locality is then guaranteed by Lemma~\ref{lem:locality}.

Finally, we examine the skip cost of $\cA$. Assume a column is erased. If that column comes from a block of $\cB'$, then there are at least $\binom{k}{t-1}\geq 2$ copies of it. By reading a non-erased copy we can recover the column easily. If the erased column comes from a block of $\cB$, assume its contents is $\set{x_1,x_2,\dots,x_k}$. Partition this set into $\ell$ sets, $P_1,P_2,\dots,P_\ell$, of size $t-1$ each, except perhaps the last one of size $r$. For each $i\in[\ell-1]$, find a block in $\cB'$ that contains $P_i$. Then find a column containing this block with the elements of $P_i$ being the first $\abs{P_i}$ elements, and read those elements from the column at zero skip cost. For $i=\ell$, if $r=t-1$, then the same argument may be used. If $r=1$, then reading a single element always has zero skip cost. Otherwise, if $2\leq r\leq t-2$, find a block of $\cB'$ containing the $r$ elements of $P_\ell$, and then a column containing this block with the first $r$ elements being those of $P_\ell$, and read them. Overall, different $i$ result in different accessed columns. The total skip cost is therefore $0$.
\end{proof}

\begin{corollary}
\label{cor:xi1}
For any fixed $2\leq t \leq k$, there exists an infinite sequence of integers, $k\leq v_1<v_2<\cdots$, such that there are CFR codes, $\cA_i$, with parameters $(N_i,k,\binom{v_i-1}{t-1}/\binom{k-1}{t-1})_{v_i}$, locality $\ell=\ceil{k/(t-1)}$, with $\cost(\cA_i)=0$ and
\[
\lim_{i\to\infty}\xi(\cA_i)=1.
\]
\end{corollary}
\begin{proof}
By Lemma~\ref{lem: Rodl}, there exists a sequence of $(t,k,v_i)$ covering designs, $(X_i,\cB_i)$, and a sequence of $(t-1,k,v_i)$ covering designs, $(X_i,\cB'_i)$, with
\begin{align}
\label{eq:cbb1}
\abs*{\cB} &= \frac{\binom{v_i}{t}}{\binom{k}{t}} (1+o(1)), &
\abs*{\cB'} &= \frac{\binom{v_i}{t-1}}{\binom{k}{t-1}} (1+o(1)).
\end{align}
Since $t$ and $k$ are constants, we have
\begin{align*}
\abs*{\cB} &= \Theta(v_i^t)(1+o(1)), &
\abs*{\cB'} &= \Theta(v_i^{t-1}) (1+o(1)),
\end{align*}
implying that
\begin{equation}
\label{eq:cbb2}
\abs{\cB'}=o(\abs{\cB}).
\end{equation}
We apply Construction~\ref{con:combination} to the pairs $(X_i,\cB_i)$ and $(X_i,\cB'_i)$. If $r=1,t-1$, by Theorem~\ref{th:combination} we obtain a sequence of $(\abs{\cB_i}+\binom{k}{t-1}\abs{\cB'_i},k,r_1)$-CFR codes, $\cA_i$. By combining~\eqref{eq:cbb1} and~\eqref{eq:cbb2}, we have
\[
\lim_{i\to\infty}\xi(\cA_i) = \lim_{i\to\infty}\frac{\abs*{\cB_i}+\binom{k}{t-1}\abs*{\cB'_i}}{\binom{v_i}{t}/\binom{k}{t}} = 1.
\]
If $2\leq r\leq t-2$ a similar argument applies.
\end{proof}

While Construction~\ref{con:combination} attains an expansion factor of $1$, it does so only asymptotically. Thus, for many practical values its performance is inferior to the trivial Construction~\ref{con:multi}, as the following example illustrates.

\begin{example}
Consider the parameters $(t,k,v)=(5,6,12)$. In order to use Construction~\ref{con:combination}, we consult the online database of covering designs~\cite{dmgordon_coverings}, finding a $(5,6,12)$ covering design $(X,\cB)$ with $\abs{\cB}=132$ blocks, and a $(4,6,12)$ covering design $(X,\cB')$ with $\abs{\cB'}=41$ blocks. Using Construction~\ref{con:combination}, we obtain a CFR code, $\cA_2$, with parameters $(N=1362,k=5,\rho=66)_{12}$. For this CFR code, the expansion factor is
\[
\xi(\cA_2) = \frac{1362}{\binom{12}{5}/\binom{6}{5}}\approx 10.32.
\]
However, a straightforward use of the $(5,6,12)$ covering design, $(X,\cB)$, with Construction~\ref{con:multi}, results in a CFR code, $\cA_1$, with parameters $(N=264,k=5,\rho=66)_{12}$, for which we get a much lower expansion factor
\[
\xi(\cA_1) = \frac{264}{\binom{12}{5}/\binom{6}{5}}=2.
\]
\end{example}

Motivated by this example, we provide a more elaborate construction of CFR codes. It is based on a method for constructing covering designs from other covering designs that is inspired by the doubling construction of Steiner quadruple systems of Hanani~\cite{hanani1960quadruple} used in~\cite{chee2024repairing} to get a zero skip-cost FR code. While the construction almost never results in an optimal asymptotic expansion factor, it out-performs Construction~\ref{con:multi} and Construction~\ref{con:combination} for many practical values.

Before proceeding, we introduce the following notation. For any vector $\bx = (x_1, \dots, x_k) \in [v]^k$ and any $a\in [v]$, we denote 
\[\abs*{\bx}_a \eqdef \abs*{\set*{i \in [k] ~:~ x_i = a}}\] 
the number of coordinates of $\bx$ containing $a$. 

\begin{construction}\label{con:recursive}
Let $(X,\cB)$ be a properly local $(t,k,v)$ covering design, $X=[v]$, $\cB\subseteq\binom{X}{k}$. Define
\begin{align}
\label{eq:qr}
q & \eqdef \ceil*{\frac{k}{t-1}} &
r & \eqdef k \umod (t-1).
\end{align}
We construct a design $(X^*,\cB^*)$, with $X^*=[v^*]$, $\cB^*\subseteq\binom{X^*}{k}$, where $v^*=qv$, in the following way.

We first define three index sets:
\begin{align*}
I_1 & \eqdef \set*{\bi\in[q]^k ~:~ \mset{\abs{\bi}_1,\dots,\abs{\bi}_q}=\mset{r,\underbrace{t-1,\dots,t-1}_{q-1}}},\\
I_2 & \eqdef \set*{\bi\in[q]^k ~:~ \mset{\abs{\bi}_1,\dots,\abs{\bi}_q}=\mset{0,t-1+r,\underbrace{t-1,\dots,t-1}_{q-2}}},\\
I_3 & \eqdef \set*{\bi\in[q]^k ~:~ \mset{\abs{\bi}_1,\dots,\abs{\bi}_q}=\mset{m,t-1+r-m,\underbrace{t-1,\dots,t-1}_{q-2}},r+1\leq m\leq \floor*{\frac{t}{q}}}.
\end{align*}
Given an integer $x\in X=[v]$, and an integer $i\in[q]$, we define
\[ \angenv*{x,i}\eqdef x+(i-1)v\in X^*=[v^*].\]
Similarly, for $U\subseteq X=[v]$, we define the translation of $U$ by $(i-1)v$ as
\[
\angenv*{U,i} \eqdef \set*{ \angenv*{u,i} ~:~ u\in U} \subseteq X^*=[v^*].
\]
We then define four sets of blocks:
\begin{align*}
\cB_1 &\eqdef \set*{\bigcup_{j=1}^{q-1}\angenv*{U,i_j}\cup\angenv{V,i_q} ~:~ U\in\binom{X}{t-1}, V\in\binom{U}{r}, \set*{i_1,\dots,i_q}=[q]}, \\
\cB_2 &\eqdef \left\{\bigcup_{j=1}^{q-2}\angenv*{U,i_j}\cup\angenv*{V,i_{q-1}}\cup\angenv*{W,i_q} ~:~ U\in\binom{X}{t-1},V\in\binom{U}{m}, \right. \\
&\qquad\qquad\qquad\qquad\quad \left. W\in\binom{U}{t-1+r-m},\set*{i_1,\dots,i_q}=[q], r+1\leq m\leq \floor*{\frac{t}{q}} \right\},\\
\cB_3 &\eqdef \Bigg\{ \set*{x_1,\dots,x_k}\in \binom{X^*}{k} ~:~ \set*{x_1\umod v,\dots, x_k\umod v}\in\cB,\\
&\qquad\qquad\qquad\qquad\qquad\qquad\quad \parenv*{\ceil*{\frac{x_1}{v}},\dots,\ceil*{\frac{x_k}{v}}}\in I_1\cup I_2 \Bigg\}\\
\cB_4 &\eqdef \Bigg\{ \set*{x_1,\dots,x_k}\in \binom{X^*}{k} ~:~ \set*{x_1\umod v,\dots, x_k\umod v}\in\cB,\\
&\qquad\qquad\qquad\qquad\qquad\qquad\quad \parenv*{\ceil*{\frac{x_1}{v}},\dots,\ceil*{\frac{x_k}{v}}}\in I_3 \Bigg\}
\end{align*}
The entire block set, $\cB^*$ is then defined by
\[
\cB^* = 
\begin{cases}
\cB_1 \cup 	\cB_3, &\floor*{\frac{t}{q}} \leq r \leq t-1,\\
\cB_1 \cup \cB_2 \cup 	\cB_3 \cup \cB_4, &1 \leq r \leq \floor*{\frac{t}{q}} -1.
\end{cases}
\]
Finally, we construct the CFR code, $\cA$, by writing the blocks of $\cB^*$ in columns, where the numbers in each column appear in ascending order.
\end{construction}

\begin{theorem}
\label{th:coveringrec}
Assume the setting of Construction~\ref{con:recursive}. Then $(X^*,\cB^*)$ is a $(t,k,qv)$ covering design. Additionally,
\begin{equation}
\label{eq:bstar}
\abs*{\cB^*} = 
\begin{cases}
c_1(t,k)\abs*{\cB} + \binom{v}{t-1}, &r = t-1,\\
c_2(t,k)\abs*{\cB} + q\binom{t-1}{r}\binom{v}{t-1}, &\floor*{\frac{t}{q}} \leq r \leq t-2,\\
c_3(t,k)\abs*{\cB} + q\binom{t-1}{r}\binom{v}{t-1} \\
\quad + \sum_{m=r+1}^{\floor{\frac{t}{q}}} q(q-1) \binom{t-1}{m}\binom{t-1}{t-1+r-m}\binom{v}{t-1}, &1 \leq r \leq \floor*{\frac{t}{q}} -1,\\
\end{cases}
\end{equation}
where
\begin{equation}
\label{eq:c123}
\begin{split}
c_1(t,k) &= \frac{k!}{(t-1)!^q} + q(q-1) \frac{k!}{(2t-2)! (t-1)!^{q-2} },\\
c_2(t,k) &= q \frac{k!}{r!(t-1)!^{q-1}} + q(q-1) \frac{k!}{(t-1+r)! (t-1)!^{q-2}},\\
c_3(t,k) &= c_2(t,k)+ q(q-1)\sum_{m=r+1}^{ \floor{\frac{t}{q}}} \frac{k!}{m!(t-1+r-m)! (t-1)!^{q-2} }.
\end{split}
\end{equation}
\end{theorem}

\begin{proof}
By construction, $\abs{X^*}=v^*=qv$, and $\cB^*\subseteq\binom{X^*}{k}$. Additionally, $\abs{\cB^*}$ follows directly from the construction by straightforward counting. To show that $(X^*,\cB^*)$ is a $(t,k,qv)$ covering design, it remains to prove that any $t$-subset, $T\in\binom{X^*}{t}$ is contained in at least one block of $\cB^*$. Consider such a $t$-subset, $T=\set{y_1,y_2,\dots,y_t}\in\binom{X^*}{t}$, and for all $j\in[t]$, write
\begin{align*}
x_j &= y_j \umod v, & i_j & = \ceil*{\frac{y_j}{v}},
\end{align*}
so
\[
y_j=\angenv*{x_j,i_j}, \ \text{ with }\  x_j\in X=[v], i_j\in[q].
\]
For convenience, we denote $\bi=(i_1,\dots,i_t)$. We distinguish between two cases.

\textbf{Case 1:} Assume $\abs{\set{x_1,\dots,x_t}}=t$. In this case, since $(X,\cB)$ is a $(t,k,v)$ covering design, there exists a block $\set{x_1,\dots,x_k}\in\cB$ that contains $x_1,\dots,x_t$. Our strategy in this case is to show that there must exist a vector $\bi^*=(i_1,\dots,i_k)\in I_1\cup I_2\cup I_3$ whose prefix is $\bi$, and define,
\[
T^*\eqdef \set*{\angenv*{x_1,i_1},\dots,\angenv*{x_k,i_k}}.
\]
It would then follow that $T\subseteq T^*\in\cB_3\cup\cB_4$, implying, depending on the parameter range, that $T\in\cB^*$. We distinguish between the following sub-cases.

If $\abs{\set{i_1,\dots,i_t}}\leq q-1$, then there exists $\bi^*\in I_2$ whose prefix is $\bi$. Then $T\subseteq T^*\in\cB_3\subseteq \cB^*$. Otherwise, $\set{i_1,\dots,i_t}=[q]$. Let $m$ be the smallest multiplicity in $i_1,\dots,i_t$, i.e.,
\begin{equation}
\label{eq:defm}
m \eqdef \min_{a\in[q]} \abs*{\bi}_a.
\end{equation}
Obviously,
\begin{equation}
\label{eq:mtq}
m \leq \floor*{\frac{t}{q}}.
\end{equation}
If $1\leq m\leq r$, then we can find $\bi^*\in I_1$ whose prefix is $\bi$, and so $T\subseteq T^*\in\cB_3\subseteq \cB^*$. Otherwise, we have $r+1\leq m\leq \floor{\frac{t}{q}}$, implying we are in the construction case of $1\leq r\leq\floor{\frac{t}{q}}-1$. We can then find $\bi^*\in I_3$ whose prefix is $\bi$, and so $T\subseteq T^*\in\cB_4\subseteq \cB^*$.

\textbf{Case 2:} Assume $\abs{\set{x_1,\dots,x_t}}\leq t-1$. In this case, we will show that $T\in\cB_1\cup\cB_2$, depending on the parameter range, and therefore, $T\in\cB^*$. We define $m$ as in~\eqref{eq:defm}, and therefore~\eqref{eq:mtq} holds as well. If $1\leq m\leq r$, then by construction, $T\in\cB_1$, and therefore $T\in\cB^*$. Otherwise, we have $r+1\leq m\leq \floor{\frac{t}{q}}$, implying we are in the construction case of $1\leq r\leq\floor{\frac{t}{q}}-1$. By construction, we then have $T\in\cB_2$, implying $T\in\cB^*$.
\end{proof}

\begin{theorem}
\label{th:cfrrec}
Let $\cA$ be the CFR code from Construction~\ref{con:recursive}. Then it is a $(\abs{\cB^*},k,\binom{qv-1}{t-1}/\binom{k-1}{t-1})_{qv}$-CFR code with locality $\ell=\ceil{k/(t-1)}$ and $\cost(\cA)=0$, where $\abs{\cB^*}$ is given in Theorem~\ref{th:coveringrec}.
\end{theorem}

\begin{proof}
By Theorem~\ref{th:coveringrec}, $(X^*,\cB^*)$ is a covering design, which immediately (together with Lemma~\ref{lem:locality}) settles all of the claimed parameters, except for the zero skip cost. To prove the latter, assume a column of $\cA$ is erased, which contains the elements of a block $E\in\cB^*$. We show how to recover the erased column with zero skip cost. We distinguish between several cases:

\textbf{Case 1:} $E\in\cB_1$. W.l.o.g., assume 
\[
E=\bigcup_{j=1}^{q-1}\angenv*{U,j} \cup \angenv*{V,q},
\]
for some $U=\set{u_1,\dots,u_{t-1}}\in\binom{X}{t-1}$ and $V=\set{u_1,\dots,u_r}$. Since $(X,\cB)$ is a $(t,k,v)$ covering design, there exist at least $r_{t-1}$ distinct blocks containing $U$. Since the design is properly local, by~\eqref{eq:proploc}, $r_{t-1}\geq q$. We therefore denote by $B_1,\dots,B_q\in\cB$, a list of $q$ distinct blocks containing $U$. We now use these blocks to find blocks in $\cB^*$.

For all $i\in[q-1]$, denote the elements of $B_i$ as $B_i=\set{u_1,\dots,u_{t-1},y_{i,t},\dots,y_{i,k}}$. We consider the set
\[
B^*_i = \set*{\angenv*{u_1,i},\dots,\angenv*{u_{t-1},i},\angenv*{y_{i,t},j_{i,t}},\dots,\angenv*{y_{i,k},j_{i,k}}},
\]
such that $(i,\dots,i,j_{i,t},\dots,j_{i,k})\in I_1$. In particular, note that $j_{i,a}\neq i$ for all $a$. Thus, when the elements of $B^*_i$ are sorted, $\angenv*{u_1,i},\dots,\angenv*{u_{t-1},i}$ appear consecutively. Additionally, by construction, $B^*_i\in\cB_3$.

Similarly, for $i=q$, denote the elements of $B_q$ as $B_q=\set{u_1,\dots,u_{r},y_{q,r+1},\dots,y_{q,k}}$. We consider the set
\[
B^*_q = \set*{\angenv*{u_1,q},\dots,\angenv*{u_{r},q},\angenv*{y_{q,r+1},j_{q,r+1}},\dots,\angenv*{y_{q,k},j_{q,k}}},
\]
such that $(q,\dots,q,j_{q,r+1},\dots,j_{q,k})\in I_1$ and $j_{q,a}\neq q$ for all $a$. Again, when the elements of $B^*_q$ are sorted, $\angenv*{u_1,q},\dots,\angenv*{u_{r},q}$ appear consecutively, and $B^*_q\in\cB_3$.

The distinctness of $B_1,\dots,B_q$ implies the distinctness of $B^*_1,\dots,B^*_q$. We then recover $\angenv{U,i}$ from $B^*_i$, for all $i\in[q-1]$, and $\angenv{V,q}$ from $B^*_q$. This is done with zero skip cost.

\textbf{Case 2:} $E\in\cB_2$. W.l.o.g., assume 
\[
E=\bigcup_{j=1}^{q-2}\angenv*{U,j} \cup \angenv*{V,q-1}\cup \angenv*{W,q},
\]
for some $U=\set{u_1,\dots,u_{t-1}}\in\binom{X}{t-1}$ and $V\in\binom{U}{m}$, $W\in\binom{U}{t-1+r-m}$, and some $r+1\leq m\leq \floor{\frac{t}{q}}$. This case is handled just like Case 1, with the two smaller sets, $\angenv{V,q-1}$ and $\angenv{W,q}$, receiving the same treatment as $\angenv{V,q}$ in Case 1.

\textbf{Case 3:} $E=\set{\angenv{x_1,i_1},\dots,\angenv{x_k,i_k}}\in\cB_3$ and $\bi=(i_1,\dots,i_k)\in I_1$. Thus, we can write $E=\bigcup_{j=1}^{q}\angenv{U_i,a_i}$, with $U_i\in\binom{X}{t-1}$ for all $i\in[q-1]$, and $U_q\in\binom{X}{r}$. Additionally, the sets $U_i$ are pairwise disjoint.

For each $i\in[q-1]$, we pick an arbitrary $V_i\in\binom{U_i}{r}$, an arbitrary $b_i\in[q]\setminus\set{a_i}$, and consider the set
\[
B^*_i = \bigcup_{j\in[q]\setminus\set{b_i}}\angenv*{U_i,j}\cup \angenv*{V_i,b_j} \in \cB_1.
\]
For $i=q$, we pick an arbitrary $V\in\binom{X}{t-1}$ such that $U_q\subseteq V$, and consider the set
\[
B^*_q = \bigcup_{j\in[q]\setminus\set{a_q}}\angenv*{V,j}\cup \angenv*{U_q,a_q} \in \cB_1.
\]
The sets $B^*_i$, $i\in[q]$ are clearly distinct, since the sets $U_i$ are pairwise disjoint. For each $i\in[q]$, we further have that $\angenv{U_i,a_i}\subseteq B^*_i$, and when the elements of $B^*_i$ are sorted, the elements of $\angenv{U_i,a_i}$ appear contiguously. Therefore, $E$ may be recovered with zero skip cost.

\textbf{Case 4:} 
$E=\set{\angenv{x_1,i_1},\dots,\angenv{x_k,i_k}}\in\cB_3$ and $\bi=(i_1,\dots,i_k)\in I_2$. Thus, we can write $E=\bigcup_{j=1}^{q-1}\angenv{U_i,a_i}$, with $U_i\in\binom{X}{t-1}$ for all $i\in[q-2]$, and $U_{q-1}\in\binom{X}{t-1+r}$. Additionally, the sets $U_i$ are pairwise disjoint. We arbitrarily partition $U_{q-1}$ into two sets, $U'_{q-1},U'_{q}\subseteq U_{q-1}$, $U'_{q-1}\cup U'_q=U_{q-1}$, $U'_{q-1}\cap U'_q=\emptyset$, $\abs{U'_{q-1}}=t-1$, and $\abs{U'_q}=r$. We similarly define $a'_{q-1}=a'_q=a_{q-1}$. The remainder of the proof for this case is the same as Case 3, where instead of $U_1,\dots,U_q$ and $a_1,\dots,a_q$ we use $U_1,\dots,U_{q-2},U'_{q-1},U'_q$ and $a_1,\dots,a_{q-2},a'_{q-1},a'_q$.

\textbf{Case 5:} 
$E=\set{\angenv{x_1,i_1},\dots,\angenv{x_k,i_k}}\in\cB_4$. Thus, we can write $E=\bigcup_{j=1}^{q}\angenv{U_i,a_i}$, with $U_i\in\binom{X}{t-1}$ for all $i\in[q-2]$, $U_{q-1}\in\binom{X}{m}$, and $U_q\in\binom{X}{t-1+r-m}$, for some $r+1\leq m\leq\floor{\frac{t}{q}}$. Additionally, the sets $U_i$ are pairwise disjoint. The proof for this case follows the exact same logic as the proof of Case 3. The $q-2$ sets, $U_1,\dots,U_{q-2}$, receive the same treatment as their equivalent in Case 3. The two smaller sets, $U_{q-1}$ and $U_q$, receive the same treatment as the small set $U_q$ in Case 3.
\end{proof}

\begin{example}
Let $(X,\cB)$ be a $(3,4,5)$ covering design with four blocks:
\[
\cB=\set*{
\set*{1,2,3,4},
\set*{1,2,3,5},
\set*{1,2,4,5},
\set*{1,3,4,5}
}.
\]
We apply Construction~\ref{con:recursive} to this design. In this case, $q=2$ and $r=2$. During the construction we obtain the following blocks:
\begin{itemize}
\item
The set $\cB_1$ contains the following $10$ blocks:
\[
\begin{array}{lllll}
\set{1,2,6,7}, & \set{1,3,6,8}, & \set{1,4,6,9}, & \set{1,5,6,10}, & \set{2,3,7,8},\\
\set{2,4,7,9}, & \set{2,5,7,10}, & \set{3,4,8,9}, & \set{3,5,8,10}, & \set{4,5,9,10}.
\end{array}
\]
\item
The set $\cB_3$ contains the following $24$ blocks resulting from index vectors from $I_1$:
\[
\begin{array}{cllll}
\toprule
\text{index} & \text{blocks} \\
\midrule
(1,1,2,2) & \set{1,2,8,9}, &\set{1,2,8,10}, &\set{1,2,9,10}, & \set{1,3,9,10},\\
(1,2,1,2) & \set{1,3,7,9}, &\set{1,3,7,10}, &\set{1,4,7,10}, &\set{1,4,8,10},\\
(1,2,2,1) & \set{1,4,7,8}, &\set{1,5,7,8}, &\set{1,5,7,9}, &\set{1,5,8,9},\\
(2,1,1,2) & \set{2,3,6,9}, &\set{2,3,6,10}, &\set{2,4,6,10}, &\set{3,4,6,10},\\
(2,1,2,1) & \set{2,4,6,8}, &\set{2,5,6,8}, &\set{2,5,6,9}, &\set{3,5,6,9},\\
(2,2,1,1) &\set{3,4,6,7}, &\set{3,5,6,7}, &\set{4,5,6,7}, &\set{4,5,6,8}.\\
\bottomrule
\end{array}
\]
\item
The set $\cB_3$ contains the following $8$ blocks resulting from index vectors from $I_2$:
\[
\begin{array}{cllll}
\toprule
\text{index} & \text{blocks} \\
\midrule
(1,1,1,1) & \set{1,2,3,4}, &\set{1,2,3,5}, &\set{1,2,4,5}, & \set{1,3,4,5},\\
(2,2,2,2) & \set{6,7,8,9}, &\set{6,7,8,10}, &\set{6,7,9,10}, &\set{6,8,9,10}.\\
\bottomrule
\end{array}
\]
\end{itemize}

All the blocks together form a $(3,4,10)$ covering design with $42$ blocks. Writing the blocks as the columns of a $4\times 42$ matrix we obtain
\[
\arraycolsep=0.5pt
\cA=\parenv*{
\begin{array}{cccccccccccccccccccccccccccccccccccccccccc}
1&1&1&1&2&2&2&3&3&4&1&1&1&1&1&1&1&1&1&1&1&1&2&2&2&3&2&2&2&3&3&3&4&4&1&1&1&1&6&6&6&6\\
2&3&4&5&3&4&5&4&5&5&2&2&2&3&3&3&4&4&4&5&5&5&3&3&4&4&4&5&5&5&4&5&5&5&2&2&2&3&7&7&7&8\\
6&6&6&6&7&7&7&8&8&9&8&8&9&9&7&7&7&8&7&7&7&8&6&6&6&6&6&6&6&6&6&6&6&6&3&3&4&4&8&8&9&9\\
7&8&9&0&8&9&0&9&0&0&9&0&0&0&9&0&0&0&8&8&9&9&9&0&0&0&8&8&9&9&7&7&7&8&4&5&5&5&9&0&0&0
\end{array}
},
\]
(where $10$ is written as $0$ to maintain a condensed presentation). The matrix $\cA$ is a $(42,4,12)_{10}$-CFR code with zero skip cost.
\end{example}

\begin{corollary}
\label{cor:recursive}
For any fixed $2\leq t \leq k$, there exists an infinite sequence of integers, $k\leq v_1<v_2<\cdots$, such that there are CFR codes, $\cA_i$, with parameters $(N_i,k,\binom{qv_i-1}{t-1}/\binom{k-1}{t-1})_{qv_i}$, locality $\ell=\ceil{k/(t-1)}$, with $\cost(\cA_i)=0$ and
\[
\lim_{i\to\infty}\xi(\cA_i)=\frac{1}{q^t}\cdot
\begin{cases}
c_1(t,k), &r = t-1,\\
c_2(t,k), &\floor*{\frac{t}{q}} \leq r \leq t-2,\\
c_3(t,k), &1 \leq r \leq \floor*{\frac{t}{q}} -1,\\
\end{cases}
\]
where $q$ and $r$ are defined in~\eqref{eq:qr}, and $c_1(t,k)$, $c_2(t,k)$ and $c_3(t,k)$ are defined in~\eqref{eq:c123}.
\end{corollary}
\begin{proof}
By Lemma~\ref{lem: Rodl}, there exists a sequence of $(t,k,v_i)$ covering designs, $(X_i,\cB_i)$, with
\[
\abs*{\cB} = \frac{\binom{v_i}{t}}{\binom{k}{t}} (1+o(1)).
\]
We apply Construction~\ref{con:recursive}, which by Theorem~\ref{th:cfrrec} gives us a sequence of $(\abs{\cB^*_i},k,\binom{qv_i-1}{t-1}/\binom{k-1}{t-1})$-CFR codes, $\cA_i$, and where $\abs{\cB^*}$ is given by~\eqref{eq:bstar} and~\eqref{eq:c123}. We therefore have
\begin{align*}
\lim_{i\to\infty}\xi(\cA_i) &= \lim_{i\to\infty}\frac{\abs*{\cB^*_i}}{\binom{qv_i}{t}/\binom{k}{t}} = \lim_{i\to\infty}\frac{\abs*{\cB_i}}{\binom{qv_i}{t}/\binom{k}{t}}\cdot\frac{\abs*{\cB^*_i}}{\abs*{\cB_i}},
\end{align*}
which proves our claim since the first fraction tends to $\frac{1}{q^t}$, and by~\eqref{eq:bstar} and~\eqref{eq:c123}, the second tends to $c_j(t,k)$, $j\in\set{1,2,3}$, where $j$ is determined by the parameters ranges.
\end{proof}

Some examples of the asymptotic expansion factor resulting from Construction~\ref{con:recursive} are provided in Table~\ref{tab:asympt}. We note that generally, the asymptotic expansion factor guaranteed by Construction~\ref{con:recursive} is worse (i.e., higher) than that of Construction~\ref{con:combination}. However, in the non-asymptotic regime, Construction~\ref{con:recursive} may significantly outperform Construction~\ref{con:combination}. As example, we consider CFR codes from $(5,6,v)$ covering designs for even values of $12\leq v\leq 26$. For the constructions we use the smallest known covering designs parameters given in~\cite{dmgordon_coverings}, summarized in Table~\ref{tab:known}. The resulting expansion factors from Constructions~\ref{con:multi},~\ref{con:combination}, and~\ref{con:recursive}, are summarized in Table~\ref{tab:comp}.

\begin{table}
\centering
\caption{The asymptotic expansion factor of CFR codes resulting from Construction~\ref{con:recursive} and Corollary~\ref{cor:recursive}}
\label{tab:asympt}
\begin{tabular}{c|cccccccccc}
\toprule
$(t,k)$ & $(3,4)$ & $(4,5)$ & $(4,6)$ & $(5,6)$ & $(5,7)$ & $(5,8)$ & $(6,7)$ & $(6,8)$  & $(6,9)$  & $(6,10)$  \\
\hline
$\xi$ & $1$ & $\frac{11}{8}$  & $\frac{11}{8}$  & $1$ & $\frac{9}{4}$ & $\frac{9}{4}$  & $\frac{57}{32}$  & $\frac{57}{32}$  & $\frac{127}{32}$  & $\frac{127}{32}$ \\
\bottomrule
\end{tabular}
\end{table}

\begin{table}
\centering
\caption{Some best known sizes of covering designs from \cite{dmgordon_coverings}}
\label{tab:known}
\begin{tabular}{l|llllllll}
\toprule
\diagbox{$(t,k)$}{$v$} & $6$ & $7$ & $8$ & $9$ & $10$ & $11$ & $12$ & $13$ \\ 
\hline
$(5,6)$          & 1      & 6      & 12     & 30     & 50     & 100     & 132    & 245    \\ 
\bottomrule
\end{tabular}
\begin{tabular}{l|llllllll}
\toprule
\diagbox{$(t,k)$}{$v$} & $12$ & $14$ & $16$ & $18$ & $20$ & $22$ & $24$ & $26$ \\ 
\hline
$(5,6)$            & 132    & 371    & 808    & 1530   & 2800   & 4659   & 7084   & 11544  \\
\hline
$(4,6)$            & 40     & 80     & 152    & 236    & 382    & 580    & 784    & 1152   \\
\bottomrule
\end{tabular}
\end{table}

\begin{table}
\centering
\caption{The expansion factor of CFR codes from $(5,6,v)$ covering designs using Constructions~\ref{con:multi},~\ref{con:combination}, and~\ref{con:recursive} with designs from Table~\ref{tab:known} }
\label{tab:comp}
\begin{tabular}{c|cccccccccc}
\toprule
$v$ & $12$ & $14$ & $16$ & $18$ & $20$ & $22$ & $24$ & $26$ \\
\hline
Construction~\ref{con:multi} & $2.00$ & $2.22$  & $2.22$  & $2.14$ & $2.17$ & $2.12$  & $2.00$  & $2.11$ \\
\hline
Construction~\ref{con:combination} & $10.32$  & $8.30$ & $7.37$ & $6.03$  & $5.52$  & $5.03$ & $4.32$ & $4.21$  \\
\hline
Construction~\ref{con:recursive} & $1.61$ & $1.83$  & $1.68$  & $1.73$ & $1.59$ & $1.63$  & $1.43$  & $1.50$  \\
\bottomrule
\end{tabular}
\end{table}

We conclude this section by showing that, for all sufficiently large $v$, any $(t,k,v)$ covering design may be directly placed as the columns of a CFR code with zero skip cost. This will create a CFR code with the optimal expansion factor. However, the claim is proved using the probabilistic method, and is therefore non-constructive.

\begin{theorem}
\label{th:prob}
Fix integers $2\leq t\leq k$. Then for all sufficiently large $v$, any $(t,k,v)$ covering design, $(X,\cB)$, can create a CFR code $\cA$, whose columns are the blocks of $\cB$, with locality $\ell=\ceil{k/(t-1)}$, such that $\cost(\cA)=0$.
\end{theorem}

\begin{proof}
Let us define
\begin{align*}
q &\eqdef \ceil*{\frac{k}{t-1}}, & r &\eqdef k \umod (t-1), & N&\eqdef \abs{\cB}.
\end{align*}
Observe that for sufficiently large $v$ the covering design is properly local. We further note that we can assume $\cB$ is a set and not a multiset, i.e., there are no repeated blocks. Indeed, if a column with a duplicate is erased, we can recover it by reading the duplicate. Thus, for the sake of the analysis to follow we remove duplicates of columns. In particular, this implies we may assume
\begin{equation}
\label{eq:Nvk}
N=\abs{\cB} \leq \binom{v}{k} \leq v^k.
\end{equation}

We place the $N$ blocks of $\cB$ in a $k\times N$ array, $\cA$ randomly: for each block $B_i=\set{x_{i,1},\dots,x_{i,k}}$, independently of other blocks, we randomly and uniformly choose a permutation $\sigma\in S_k$ (the set of permutations over $[k]$), and place the elements in the column of $\cA$ in the order $x_{i,\sigma(1)},\dots,x_{i,\sigma(k)}$. By showing that, with positive probability, all columns of $\cA$ may be recovered with zero skip cost, we will deduce the claim.

Consider column $i$ of $\cA$, corresponding to the block $B_i=\set{x_{i,1},\dots,x_{i,k}}\in\cB$. We further partition $B_i$ into $q$ sets, each of size $t-1$, except perhaps the last one of size $r$:
\[
B_{i,1}=\set{x_{i,1},\dots,x_{i,t-1}}, B_{i,2}=\set{x_{i,t},\dots,x_{i,2(t-1)}},\dots,B_{i,q}=\set{x_{i,(q-1)(t-1)},\dots,x_{i,k}}.
\]
For each $j\in[q]$, there exist at least $r_{t-1}=\frac{v-t+1}{k-t+1}$ blocks of $\cB$ that contain $B_{i,j}$. Thus, we can find $q$ pair-wise disjoint sets of blocks, $\cH_{i,1},\dots,\cH_{i,q}\subseteq \cB\setminus\set{B_i}$, such that for all $j\in[q]$, and all $H\in\cH_{i,j}$, we have $B_{i,j}\subseteq H$, and 
\[
\abs*{\cH_{i,j}}=\floor*{\frac{r_{t-1}-1}{q}}=\floor*{\frac{v-k}{q(k-t+1)}}.
\]
For all sufficiently large $v$ we therefore have $\abs{\cH_{i,j}}\geq 1$.

Consider a set $B_{i,j}$ and a column made of the elements of $H$, such that $B_{i,j}\subseteq H$. For now, assume $\abs{B_{i,j}}=t-1$. Out of the $k!$ possible arrangements of the elements of $H$ in the column, exactly $(k-t+2)(t-1)!(k-t+1)!=(t-1)!(k-t+2)!$ of them contain the elements of $B_{i,j}$ in a contiguous block, thus allowing their retrieval with zero skip cost.

Let us denote by $E_{i,j}$ the event that $B_{i,j}$ may be recovered with zero skip cost. Then
\begin{equation}
\label{eq:bijtemp}
\Pr\sparenv*{E_{i,j}} \geq 1-\parenv*{1-\frac{(t-1)!(k-t+2)!}{k!}}^{\floor{(v-k)/q(k-t+1)}},
\end{equation}
since we have at least $\floor{r_{t-1}/q}$ columns in $\cH_{i,j}$ containing the elements of $B_{i,j}$. Returning to our assumption, if $1\leq B_{i,j} < t-1$, then it is easily checked that~\eqref{eq:bijtemp} holds as well. We now conveniently denote
\[
\alpha_{t,k} \eqdef 1-\frac{(t-1)!(k-t+2)!}{k!}
\]
and note that $0<\alpha_{t,k}<1$ is a constant (depending on $t$ and $k$). Then
\[
\beta_{t,k}\eqdef \alpha_{t,k}^{\frac{1}{2k^2}}
\]
also is a constant that satisfies $0<\beta_{t,k}<1$. We can now develop~\eqref{eq:bijtemp} to get
\begin{equation}
\label{eq:bij}
\Pr[E_{i,j}] \geq 1-\parenv*{1-\frac{(t-1)!(k-t+2)!}{k!}}^{\floor{(v-k)/q(k-t+1)}} \geq 1-\beta_{t,k}^v
\end{equation}
since for all large enough $v$ we have $\floor{(v-k)/q(k-t+1)}\geq \frac{v}{2k^2}$.

Let $E_i$ be the event that column $i$ of $\cA$, corresponding to block $B_i$, may be recovered with zero skip cost. Using~\eqref{eq:bij}, we have
\begin{align}
\Pr[E_i] &\overset{(a)}{\geq} \parenv*{1-\beta_{t,k}^v}^q 
\overset{(b)}{\geq} 1-q\beta_{t,k}^v,
\label{eq:prei}
\end{align}
where $(a)$ holds since the sets $\cH_{i,j}$ are pair-wise disjoint, and so their columns are statistically independent, and $(b)$ follows from Bernoulli's inequality $(1-x)^a \geq 1-ax$ for any real $x\in[0,1]$ and $a\geq 1$.

We now have
\begin{align*}
\Pr[\cost(\cA)=0] &= \Pr\sparenv*{\bigcap_{i\in[N]}E_i} = 1-\Pr\sparenv*{\bigcup_{i\in[N]}\overline{E_i}} \overset{(a)}{\geq} 1-\sum_{i\in[N]}\Pr\sparenv*{\,\overline{E_i}\,}\\
& \overset{(b)}{\geq} 1-Nq\beta_{t,k}^v 
\overset{(c)}{\geq} 1-v^k q\beta_{t,k}^v,
\end{align*}
where $(a)$ follows from the union bound, $(b)$ follows from~\eqref{eq:prei}, and $(c)$ follows from~\eqref{eq:Nvk}. We now crucially note that $q$, $t$, $k$ and $\beta_{t,k}$ are constants. Thus, $v^k q$ grows polynomially in $v$, whereas $\beta_{t,k}^v$ vanishes exponentially. Thus, for all sufficiently large $v$ we have
\[
\Pr[\cost(\cA)=0]>0,
\]
and a CFR code with zero skip cost exists as claimed.
\end{proof}

\section{Conclusion}
\label{sec:conc}

In this work we studied constructions of fractional-repetition codes based on covering designs (CFR codes). We additionally required the codes to have zero skip cost when recovering from a column erasure. We showed this is always attainable, though at a price: we may need to increase the number of columns with respect to the optimal number determined by Steiner systems. We called this price the expansion factor.

We showed how zero skip cost CFR codes with asymptotic expansion factor $2$ and $1$ are constructed (Constructions~\ref{con:multi} and~\ref{con:combination}). However, until reaching this asymptote, another construction we presented obtained a lower expansion factor (Construction~\ref{con:recursive}). We also showed any covering design may be arranged in a CFR code to get a zero skip cost, albeit, the proof was non-constructive (Theorem~\ref{th:prob}). Some open questions remain, with the most important one being: Is it possible to efficiently find the permutation guaranteed by Theorem~\ref{th:prob}? If this is hard to find, can we find a better construction than Construction~\ref{con:recursive} in the non-asymptotic regime?

\bmhead*{Acknowledgments}
The work of Bo-Jun Yuan was supported by National Natural Science Foundation of China under Grant 12201559.
\bibliography{reference}

\end{document}